\documentclass[runningheads]{llncs}
\usepackage{amssymb,amsfonts,amsmath}
\usepackage{cite}
\usepackage{graphicx}
\usepackage{array}
\usepackage{algorithm}
\usepackage{algorithmicx}
\usepackage[noend]{algpseudocode}
\usepackage[verbose]{wrapfig}
\usepackage[colorlinks,linkcolor=cyan,citecolor=blue]{hyperref}
\usepackage{multirow}
\def\cL{{\cal L}}
\def\cT{{\cal T}}

\def\pari{\mathsf{Parikh}}
\def\cov{\mathsf{cover}}
\def\bin{\mathsf{bin}}
\def\free{\mathit{free}}
\def\lef{\mathit{left}}
\def\Set{{\sf Set}}
\def\ml{{\sf ml}}

\def\AF{{\sf AF}}
\def\RT{{\sf RT}}
\def\ART{{\sf ART}}

\def\A4F{{\sf A4F}}

\begin{document}

\title{Abelian Repetition Threshold Revisited}
\author{Elena A. Petrova \and Arseny M. Shur}
\institute{Ural Federal University, Ekaterinburg, Russia\\ \{elena.petrova,arseny.shur\}@urfu.ru}
\maketitle

\begin{abstract}
Abelian repetition threshold $\ART(k)$ is the number separating fractional Abelian powers which are avoidable and unavoidable over the $k$-letter alphabet. The exact values of $\ART(k)$ are unknown; the lower bounds were proved in [A.V. Samsonov, A.M. Shur. On Abelian repetition threshold. RAIRO ITA, 2012] and conjectured to be tight. We present a method of study of Abelian power-free languages using random walks in prefix trees and some experimental results obtained by this method. On the base of these results, we conjecture that the lower bounds for $\ART(k)$ by Samsonov and Shur are not tight for all $k\ne 5$ and prove this conjecture for $k=6,7,8,9,10$. Namely, we show that $\ART(k)>\frac{k-2}{k-3}$ in all these cases.

\keywords{Abelian-power-free language, repetition threshold, prefix tree, random walk}
\end{abstract}

\section{Introduction}
Two words are \emph{Abelian equivalent} if they have the same multiset of letters (in other  terms, if they are anagrams of each other); \emph{Abelian repetition} is a pair of Abelian equivalent factors in a word. The study of Abelian repetitions originated from the question of Erd\"os \cite{Erd61}: does there exist an infinite finitary word having no consecutive pair of Abelian equivalent factors? The factors of the form $u\bar u$, where $u$ and $\bar u$ are Abelian equivalent, are now called \emph{Abelian squares}. In modern terms, Erd\"os's question looks like ``are Abelian squares avoidable over some finite alphabet?" This question was answered in the affirmative by Evdokimov \cite{Evd68}; the smallest possible alphabet has cardinality 4, as was proved by Ker\"anen \cite{Ker92}.  In a similar way, Abelian $k$th powers are defined for arbitrary $k\ge 2$. Dekking \cite{Dek79} constructed infinite ternary words without Abelian cubes and infinite binary words without Abelian 4th powers. The results by Dekking and Ker\"anen form an Abelian analog of the seminal result by Thue \cite{Thue06}: there exist an infinite ternary word containing no squares (factors of the form $uu$) and an infinite binary word containing no cubes (factors of the form $uuu$). 

Integral powers of words were later generalized to rational (fractional) powers: given a word $u$ of length $n$, take a length-$m$ prefix $v$ of the infinite word $uuu\cdots$; then $v$ is the $(\frac{m}{n})$th power of $u$ ($m>n$ is assumed). Usually, $v$ is referred to as a $(\frac{m}{n})$-power. A word is said to be \emph{$\alpha$-power-free} if it contains no $(\frac{m}{n})$-powers with $\frac{m}{n}\ge\alpha$. Fractional powers gave rise to the notion of \emph{repetition threshold} which is the function $\RT(k)=\inf\{\alpha:\text{there exists an infinite $k$-ary $\alpha$-power-free word}\}$. The value $\RT(2)=2$ is known since Thue \cite{Thue12}. Dejean \cite{Dej72} showed that $\RT(3)=7/4$ and conjectured the remaining values $\RT(4)=7/5$ (proved by Pansiot \cite{Pan84c}) and $\RT(k)=\frac{k}{k-1}$ for $k\ge 5$ (proved by efforts of many authors \cite{Mou92,Car07,CuRa11,Rao11}). An extension of the notions of fractional power and repetition threshold to the Abelian case was proposed by Samsonov and Shur \cite{SaSh12}. Integral Abelian powers can be generalized to fractional ones in several ways; however, for the case $m\le 2n$, one definition of an Abelian $(\frac{m}{n})$-power is preferable due to its symmetric nature. 
According to this definition, a word $vu\bar v$ is an Abelian $(\frac{m}{n})$th power of the word $vu$ if $|vu|=n$, $|vu\bar v|=m$, and $\bar v$ is Abelian equivalent to $v$ (in \cite{SaSh12}, such Abelian powers were called \emph{strong}). Note that the reversal of an Abelian $(\frac{m}{n})$-power is also an Abelian $(\frac{m}{n})$-power. In this paper, we consider \emph{only} strong Abelian powers; see Section~\ref{s:pre} for the definition in the case $m>2n$.
Given the definition of fractional Abelian powers, one naturally defines Abelian $\alpha$-power-free words and Abelian repetition threshold $\ART(k)=\inf\{\alpha:\text{there exists an infinite $k$-ary Abelian $\alpha$-power-free word}\}$. Cassaigne and Currie \cite{CaCu99} showed that for any $\varepsilon >0$ their exists an Abelian $(1+\varepsilon)$-power-free word over a finite alphabet of size $2^{\mathrm{poly}(\varepsilon^{-1})}$. Surely, this bound is very loose but it proves that $\lim_{k\to\infty}\ART(k)=1$. In \cite{SaSh12}, the lower bounds $\ART(4)\ge 9/5$, $\ART(k)\ge\frac{k-2}{k-3}$ for $k\ge5$ were proved and conjectured to be tight; in full, this conjecture is as follows.

\begin{conjecture}[\!\!\cite{SaSh12}] \label{c:artold}
$\ART(2)=11/3$; $\ART(3)=2$; $\ART(4)=9/5$;  $\ART(k)=\frac{k-2}{k-3}$ for $k\ge 5$.
\end{conjecture}

Up to now, no exact values of $\ART(k)$ are known. One reason for the lack of progress in estimating $\ART(k)$ is the fact that the number of $k$-ary Abelian $\alpha$-power-free words can be finite but so huge that these words cannot be enumerated by exhaustive search. 

In the present study, we approach Abelian $\alpha$-power-free words using randomized depth-first search. The language $\AF(k,\alpha)$ of all $k$-ary Abelian $\alpha$-power-free words is viewed as a \emph{prefix tree} $\cT_{k,\alpha}$: the elements of $\AF(k,\alpha)$ are nodes of the tree and $u$ is an ancestor of $v$ in the tree iff $u$ is a prefix of $w$. The search starts at the root (empty word) and is organized as follows. Reaching a node $u$ for the first time, we choose a random letter $a$, check \emph{ad hoc} whether $ua\in \AF(k,\alpha)$ and descend to $ua$ if this node exists; visiting $u$ next times, we choose a random letter among the letters not chosen at $u$ before, and proceed the same way. If there is no choice, we return to the parent of $u$ (and thus will never reach $u$ in the future). We repeated the search multiple times and analysed the maximum level of a node reached in the tree and the change of level during the search. Based on this analysis, we state 

\begin{conjecture} \label{c:art}
$\ART(2)>11/3$; $2<\ART(3)\le 5/2$; $\ART(4)>9/5$; $\ART(5)=3/2$; $4/3<\ART(6)<3/2$; $\ART(k)=\frac{k-3}{k-4}$ for $k\ge 7$.
\end{conjecture}

As a first step in proving this conjecture, we prove $\ART(k)>\frac{k-2}{k-3}$ for $k=6,7,8,9,10$ by some exhaustive search  (Theorem~\ref{t:678910} in Section~\ref{s:exhaust}). 

The paper is organized as follows. After preliminary Section~\ref{s:pre}, we describe the algorithms used in our study and the results obtained through experiments in Sections~\ref{s:rand} and~\ref{s:exhaust} respectively. Section~\ref{s:fin} contains some final remarks and prospects for future studies.

\section{Definitions and Notation} \label{s:pre}

We study finite words over finite alphabets, using standard notation $\Sigma$ for a (linearly ordered) alphabet, $\sigma$ for its size, $\Sigma^*$ for the set of all finite words over $\Sigma$, including the empty word $\lambda$. For a length-$n$ word $u\in\Sigma^*$ we write $u=u[1..n]$; the elements of the range $[1..n]$ are \emph{positions} in $u$, the length of $u$ is denoted by $|u|$. A word $w$ is a \emph{factor} of $u$ if $u=vwz$ for some (possibly empty) words $v$ and $z$; the condition $v=\lambda$ (resp., $z=\lambda$) means that $w$ is a \emph{prefix} (resp., \emph{suffix}) of $u$. Any factor $w$ of $u$ can be represented as $w=u[i..j]$ for some $i$ and $j$ ($j<i$ means $w=\lambda$). A factor $w$ of $u$ can have several such representations; we say that $su[i..j]$ specifies the \emph{occurrence of $w$ at position $i$}.

A \emph{$k$-power} of a word $u$ is the concatenation of $k$ copies of $u$, denoted by $u^k$. This notion can be extended to \emph{$\alpha$-powers} for an arbitrary rational $\alpha>1$. The $\alpha$-power of $u$ is the word $u^{\alpha}=u\cdots uu'$ such that $|u^\alpha|=\alpha|u|$ and $u'$ is a prefix of $u$. A word is \emph{$\alpha$-free} (resp., \emph{$\alpha^+\!$-free}) if no one of its factors is a $\beta$-power with $\beta\ge \alpha$ (resp., $\beta> \alpha$).

The \emph{Parikh vector} $\Psi(u)$ of a word $u\in\Sigma^*$ is an integer vector of length $\sigma$ whose coordinates are the numbers of occurrences of the letters from $\Sigma$ in $u$. Thus, the word $acabac$ over the alphabet $\Sigma=\{a<b<c<d\}$ has the Parikh vector $(3,1,2,0)$. Two words $u$ and $v$ are \emph{Abelian equivalent} (denoted by $u\sim v$) if  $\Psi(u)=\Psi(v)$. The \emph{reversal} of a word $u=u[1..n]$ is the word $u^R=u[n]u[n{-}1]\cdots u[1]$. Clearly, $u\sim u^R$. A nonempty word $u$ is an \emph{Abelian $k$th power} (\emph{$k$-A-power}) if $u=w_1\cdots w_k$, where $w_i\sim w_j$ for all indices $i,j$. A 2-A-power is an \emph{Abelian square}, and a 3-A-power is an \emph{Abelian cube}. Thus, $k$-A-powers generalize $k$-powers by relaxing the equality of factors to their Abelian equivalence. However, there are many ways to generalize the notion of an $\alpha$-power to the Abelian case, and all of them have drawbacks. The reason is that $u\sim v$ does not imply $u[i..j]\sim v[i..j]$ for any pair of factors of $u$ and $v$. If $1<\alpha\le 2$, we define an $\alpha$-A-power as a word $vuv'$ such that $\frac{|vuv'|}{|vu|}=\alpha$ and $v\sim v'$. The advantage of this definition is that the reversal of an $\alpha$-A-power is an $\alpha$-A-power as well. For $\alpha>2$ the situation is worse: all natural definitions compatible with the definition of $k$-A-power are not symmetric with respect to reversals (see \cite{SaSh12} for more details). So we give the definition which is compatible with the case $\alpha\le 2$: an $\alpha$-A-power is a word $u_1\cdots u_ku'$ such that $\frac{|u_1\cdots u_ku'|}{|u_1|}=\alpha$, $k=\lfloor \alpha\rfloor$, $u_1\sim\cdots \sim u_k$, and $u'$ is Abelian equivalent to a prefix of $u_1$. In \cite{SaSh12}, such words are called \emph{strong} Abelian $\alpha$-powers. For a given $\alpha$, \emph{$\alpha$-A-free} and \emph{$\alpha^+\!$-A-free} words are defined in the same way as $\alpha$-free ($\alpha^+\!$-free) words. It is convenient to extend rational numbers with ``numbers'' of the form $\alpha^+$, postulating the equivalence of the inequalities $\beta>\alpha$ and $\beta\ge \alpha^+$ (resp., $\beta\le \alpha$ and $\beta<\alpha^+$).

A \emph{language} is any subset of $\Sigma^*$. The \emph{reversal} $L^R$ of a language $L$ consists of the reversals of all words in $L$. The $\alpha$-A-free language over $\Sigma$ (where $\alpha$ belongs to extended rationals) consists of all $\alpha$-A-free words over $\Sigma$ and is denoted by $\AF(\sigma,\alpha)$. These languages are the main objects of our study aimed at finding the \emph{Abelian repetition threshold}, which is the function $\ART(k)=\inf\{\alpha: \AF(k,\alpha) \text{ is infinite}\}$. The languages $\AF(\sigma,\alpha)$ are closed under permutations: if $\pi$ is a permutation of the alphabet, then the words $u$ and $\pi(u)$
are $\alpha$-A-free for exactly the same values of $\alpha$. This makes possible the enumeration of the words in languages $\AF(\sigma,\alpha)$ by considering only \emph{lexmin} words: a word $u\in\Sigma^*$ is lexmin if $u<\pi(u)$ for any permutation $\pi$ of $\Sigma$.

Suppose that a language $L$ is \emph{factorial} (i.e., closed under taking factors of its words); for example, all languages $\AF(\sigma,\alpha)$ are factorial. Then $L$ can be represented by its \emph{prefix tree} $\cT_L$, which is a rooted labeled tree whose nodes are elements of $L$ and edges have the form $u\xrightarrow{a} ua$, where $a$ is a letter. Thus $u$ is an ancestor of $v$ iff $u$ is a prefix of $v$. We study the languages $\AF(\sigma,\alpha)$ through different types of search in their prefix trees.

\section{Algorithms} \label{s:rand}

In this section we present the algorithms we develop for the use in experiments. First we describe the random depth-first search in the prefix tree $\cT=\cT_L$ of an arbitrary factorial language $L$. Given a number $N$, the algorithm visits $N$ distinct nodes of $\cT$ following the depth-first order and returns the maximum level of a visited node. The search can be easily augmented to return the word corresponding to the node of maximum level, or to log the sequence of levels of visited nodes. Algorithm~\ref{alg:dfs} below describes one iteration of the search. In the algorithm, $u=u[1..n]$ is the word corresponding to the current node; $\Set[u]$ is the set of all letters $a$ such that the search has not tried the node $ua$ yet; $\ml$ is the maximum level reached so far; $\mathsf{count}$ is the number of visited nodes; $\cL(u)$ is the predicate returning $\mathsf{true}$ if $u\in L$ and $\mathsf{false}$ otherwise. The lines 3 and 8 refer to the updates of data structures used to compute $\cL(u)$. The search starts with $u=\lambda$, $\ml=0$, $\mathsf{count}=1$. A variant of this search algorithm was used in \cite{PeSh21dlt} to numerically estimate the entropy of some $\alpha$-free and $\alpha$-A-free languages.

\begin{algorithm*}[!htb]
\caption{Random depth-first search in $\cT(L)$: one iteration} 
\label{alg:dfs}
\begin{algorithmic}[1]
    \If {$\mathsf{count}=N$} break\Comment{search finished}
    \EndIf
    \If {$\Set[u]=\varnothing$} \Comment{all children of $u$ were visited}
        \State{[update data structures]}
        \State{$u\gets u[1..|u|{-}1]$ \Comment{return to the parent of $u$}}
    \Else \State{$a\gets$ random($\Set[u]$);  $\Set[u]\gets\Set[u] - a$ \Comment{take random unused letter}}
        \If {$\cL(ua)$} \Comment{the node $ua$ is in $\cT(L)$}
            \State{[update data structures]}
            \State{$u\gets ua$; $\Set[u]\gets \Sigma$; $\mathsf{count}\gets\mathsf{count}+1$ \Comment{visit $ua$ next}}
            \If {$|u|>\ml$} $\ml\gets|u|$\Comment{update the maximum level}
            \EndIf
        \EndIf
    \EndIf
\end{algorithmic}
\end{algorithm*}

The key to an efficient search is a fast algorithm computing the predicate $\cL(u)$. The following fact is very useful: if $u'$ is a proper prefix of $u$, then the node for $u'$ exists and hence $\cL(u')=\mathsf{true}$. Below we present four algorithms checking, for a given $\alpha$, whether a word has a $\beta$-A-power with $\beta\ge\alpha$ as a suffix. The cases $\alpha<2$ and $\alpha>2$ are considered in Sections~\ref{ss:small} and~\ref{ss:big} respectively. 

\subsection{Avoiding small powers} \label{ss:small}
Let $\alpha<2$ and $u$ be a word of length $n$ such that all proper prefixes of $u$ are $\alpha$-A-free. To prove $u$  $\alpha$-A-free, it is necessary and sufficient to show that 
\begin{itemize}
    \item[$(\star)$] no suffix of $u$ can be written as $xyz$ such that $|z|>0$, $x\sim z$, and $\frac{|xyz|}{|xy|}\ge\alpha$.
\end{itemize}
\begin{remark} \label{r:n2}
Since Abelian equivalence is not closed under taking any sort of factors of words, the ratio $\frac{|xyz|}{|xy|}$ in $(\star)$ can significantly exceed $\alpha$. For example, all proper prefixes, and even all proper factors, of the word $u=abcde\,bdaec$ are  $\frac32$-A-free, while $u$ is an Abelian square. Hence for each suffix $z$ of $u$ one should check multiple candidates to the factor $x$ in $(\star)$. The number of such candidates can be as big as $\Theta(n)$; in total, $\Theta(n^2)$ candidates to the pair $(x,z)$ should be analysed.
\end{remark}

A reasonable approach is to store the Parikh vectors of all prefixes of $u$; they spend $O(n\sigma)$ words of space in total and require $O(\sigma)$ time for update when a letter is appended or deleted on the right end of the word. Then the Parikh vector of each factor of $u$ is obtained as the difference of Parikh vectors of corresponding prefixes. So one comparison of factors takes $O(\sigma)$ time, which means $\Theta(n^2\sigma)$ time for performing all comparisons of candidate pairs $(x,z)$ in $(\star)$ in a naive way (see Remark~\ref{r:n2}). Algorithm~\ref{alg:asmall} below gets many of the comparisons for free. It makes use of two length-$n$ arrays for each letter $a\in\Sigma$: $c_a[i]$ is the number of occurrences of $a$ in $u[1..i]$ (= a coordinate of the Parikh vector of the prefix $u[1..i]$) and $d_a[i]$ is the position of $i$th from the left letter $a$ in the current word $u$. Each of the arrays can be updated in $O(\sigma)$ time when a letter is appended or deleted on the right end of $u$. We specify the lines 3 and 8 of Algorithm~\ref{alg:dfs} as follows. At line 3, we delete the Parikh vector of $u$ from $c$-arrays and delete the last element of $d_b$, where $b$ is the last letter of $u$. At line 8, we add the Parikh vector of $ua$ to $c$-arrays and add a new element $|ua|$ to $d_a$.

The arrays $c_a$ and $d_a$ are used to compute two auxiliary functions, $\mathsf{Parikh}(i,j)$ and $\mathsf{cover}(\vec P, j)$. The function $\mathsf{Parikh}(i,j)$ returns the Parikh vector of $u[i..j]$; its coordinates are just the differences of the form $c_a[j]-c_a[i-1]$. The function $\mathsf{cover}(\vec P,j)$ returns the biggest number $i$ such that $\Psi(u[i..j])\ge \vec P$ or zero if there is no such number. Thus the function returns 0 if $u[1..j]$ contains, for some $a$, less $a$'s than $\vec P(a)$; i.e., if $c_a[j]<\vec P[a]$. If no such $a$ exists, then $i=\min_{a\in\Sigma}\{d_a[c_a[j]-\vec P[a]+1]\}$.
 
\begin{algorithm*}[!htb]
\caption{Abelian powers detection (case $\alpha<2$)} 
\label{alg:asmall}
\begin{algorithmic}[1]
\State{\textbf{function} $\mathsf{alphafree}(u)$ \Comment{$u$=word; $n=|u|$}}
\State{$\free\gets \mathsf{true}$ \Comment{$\alpha$-A-freeness flag}}
\For{$i=n$ downto 1+$\lceil n/2\rceil$} \Comment{$z=u[i..n]$}
    \State{$right\gets i-1$}
    \State{$len\gets n-i+1$; $\vec P\gets \pari(i,n)$ \Comment{length and Parikh vector of $z$}}
    \State{$\lef\gets \cov(\vec P, right)$}
    \If{$\lef=0$} break \Comment{$\vec P(u[1..right])\not\ge \vec P(z)$}
    \EndIf
    \While {$\lef\ge \max\{1, \lceil\frac{\alpha i-1-n}{\alpha-1}\rceil\}$}\Comment{guarantees $\frac{|xyz|}{|xy|}\ge\alpha$}
        \If{$right-\lef+1=len$} \Comment{$x=u[\lef..right]\sim z$}
            \State{$\free\gets \mathsf{false}$; break}
        \Else \Comment{shift $right$ leftwards, skip redundant comparisons}
            \State{$right\gets \lef+len-1$}
        \EndIf
        \State{$\lef\gets \cov(\vec P, right)$}
   \EndWhile
    \If{$\free=\textsf{false}$} break
    \EndIf
\EndFor
\State {return $\free$ \Comment{the answer to ``is $u$ $\alpha$-A-free?''}}
\end{algorithmic}
\end{algorithm*}

\begin{proposition} \label{p:corr-asmall}
Let $\alpha$ be a number such that $1<\alpha<2$ and $u$ be a word all proper prefixes of which are $\alpha$-A-free. Then Algorithm~\ref{alg:asmall} correctly detects whether $u$ is $\alpha$-A-free.
\end{proposition}
\begin{proof}
Let us show that Algorithm~\ref{alg:asmall} verifies the condition $(\star)$. The outer cycle of the algorithm fixes the first position $i$ of the suffix $z$ of $u$; the suffixes are analysed in the order of increased length $len=|z|$. If a forbidden suffix $xyz$ is detected during the iteration (we discuss the correctness and completeness of this detection below), then the algorithm breaks the outer cycle in line 14 and returns $\mathsf{false}$. Thus at the current iteration of the outer cycle the condition $(\star)$ is already verified for all shorter suffixes. The iteration uses a simple observation: if $x\sim z$, then every word $v$, containing $x$, satisfies $\Psi(v)\ge \Psi(z)$. We proceed as follows. We fix the rightmost position $right$ where a factor $x$ satisfying $x\sim z$ can end. Initially $right=i-1$ as $x$ can immediately precede $z$ (see the example in Remark~\ref{r:n2}). Then we compute the shortest factor $v=u[\lef..right]$ such that $\Psi(v)\ge \Psi(z)$. If $v=x$, the suffix $xz$ of $u$ violates $(\star)$. Otherwise $x$ cannot begin later than at the position $\lef$ by the construction of $v$. Hence we decrease $right$ by setting $right=\lef+|z|-1$ and repeat the above procedure in a loop. The verification of $(\star)$ for $z$ ends successfully when either $v$ does not exist (i.e., $\Psi(u[1..right])\not\ge \Psi(z)$ for the current value of $right$) or $\lef$ is too small (i.e., $xyz=u[\lef..n]$ with $|x|=|z|$ means $\frac{|xyz|}{|xy|}<\alpha$). The described process is illustrated by Fig.~\ref{f:jumping}. Details are provided below.

\begin{figure}[!htb]
\centering
    \includegraphics[scale=0.76, trim = 50 759 320 38, clip]{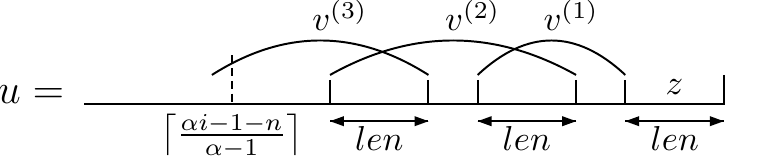}
    \includegraphics[scale=0.76, trim = 40 759 330 38, clip]{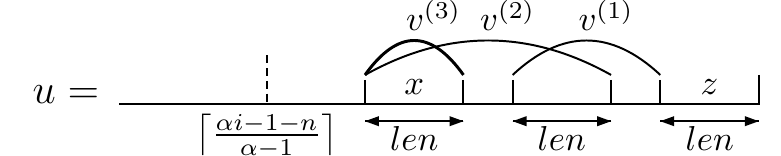}
    \caption{Illustrating the proof of Proposition~\ref{p:corr-adict}. Processing the suffix $z$, Algorithm~\ref{alg:asmall} successively finds three words $v$ satisfying $\Psi(v)\ge \Psi(z)$. On the left picture, the position $\lef$ of $v^{(3)}$ is smaller than the bound in line 8, so the verification of $(\star)$ for $z$ is finished. On the right picture, $|v^{(3)}|=|z|$, so a forbidden suffix, starting with $v^{(3)}$, is detected.}
\label{f:jumping}
\end{figure}

In lines 4--6 the algorithm calls $\pari$ to compute $\Psi(z)$ and $\cov$ to find the mentioned factor $v=u[\lef..right]$ for $right=i-1$. If $\lef=0$, then $\Psi(u[1..i{-}1])\not\ge \Psi(u[i..n])$ and hence every suffix $xyz$ of $u$ satisfies $x\not\sim z$. Moreover, one can observe that $\Psi(u[1..j{-}1])\not\ge \Psi(u[j..n])$ for each $j<i$ which immediately verifies $(\star)$ for all suffixes of $u$ that are longer than $z$. Hence in this case the verification of $(\star)$ is already done; respectively, the algorithm breaks the outer cycle in line 7 and returns $\mathsf{true}$.

If no break has happened, the algorithm enters the inner cycle, checks whether $v=x$ (line 9) and breaks with the output $\mathsf{false}$ if this condition holds. If it does not, the algorithm decreases $right$ as described above (line 12) and computes the new factor $v$ (line 13). If $v$ does not exist, $\lef$ gets 0, which results in the immediate exit from the inner cycle. If $v$ is computed but its position is too small, then the cycle is also exited. The exit means the end of the $i$th iteration.

Thus, Algorithm~\ref{alg:asmall} returns $\mathsf{false}$ only if it finds a suffix $xyz$ of $u$ which violates $(\star)$. For the other direction, let $xyz=u[j..n]$ violate $(\star)$ such that $|z|$ is minimal over all suffixes violating it. Then Algorithm~\ref{alg:asmall} cannot stop before the iteration which checks $z$. During this iteration, $\lef$ cannot become smaller than $j$ by the definition of the factor $v$. As $right$ decreases at each iteration of the inner cycle, eventually $x$ will be found. Thus the algorithm indeed verifies $(\star)$ and thus detects the $\alpha$-A-freeness of $u$. \qed
\end{proof}
\begin{remark} \label{r:time1}
As both $\pari$ and $\cov$ work in $O(\sigma)$ time, Algorithm~\ref{alg:asmall} processes a word $u$ of length $n$ in $O((K+n)\sigma)$ time, where $K$ is the total number of the inner cycle iterations during the course of the algorithm. Clearly, $K=O(n^2)$. The results of experiments suggest that $K=\Theta(n^{3/2})$ on expectation. This is indeed the case if $u$ is a random word, as Lemma~\ref{l:prob} below shows. This lemma implies that the iteration processing a suffix $z$ builds, on expectation, $O(\sqrt{|z|})$ words $v$.  
\end{remark}

\begin{lemma} \label{l:prob}
Suppose that an infinite word $\mathbf{w}$ is chosen uniformly at random among all $\sigma$-ary infinite words, $z$ is a prefix of $\mathbf{w}$, and $v$ is the shortest prefix of $\mathbf{w}[|z|{+}1..\infty]$ such that $\Psi(v)\ge \Psi(z)$. Then the expected length of $v$ is $|z|+\Omega(\sqrt{|z|})$.
\end{lemma}
\begin{proof}
Let $\ell=|z|$, $\delta=|v|-|z|$. First consider the case $\Sigma=\{0,1\}$. The process can be viewed as follows: first, $z$ is generated by $\ell$ tosses of a fair coin; then another $\ell$ tosses generate some prefix $x$ of $v$; additional tosses are made one by one until the desired result $\Psi(v)\ge \Psi(z)$ is reached after $\delta$ tosses. The Parikh vector of a word of a known length over $\Sigma$ is determined by the number of 1's. Hence $\Psi(z)$ is a random variable $\xi$ with the binomial distribution $\bin(\ell,\frac{1}{2})$; similarly, $\Psi(x)$ is a random variable $\eta$ with the same distribution.

The vector $\Psi(z)-\Psi(x)$ has the form $(-m,m)$ for some integer $m$. To obtain $v$, we should make $|m|$ ``successful'' tosses with the probability of ``success'' being $1/2$; hence the expectation of $\delta$ equals $2|m|$. Thus it remains to find the expectation of $|m|=|\xi-\eta|$. Since $E(\xi-\eta)=0$, we see that $E(|\xi-\eta|)$ is the standard deviation of $\xi-\eta$ by definition.

Due to symmetry, $\eta$ and $\ell-\eta$ have the same distribution. Hence we can replace $\xi-\eta$ by $\xi+\eta-\ell$. The random variable $\xi+\eta$ has the binomial distribution $\bin(2\ell,\frac{1}{2})$, so its standard deviation is $\sqrt{\ell/2}$. Thus $E(\delta)=2E(|\xi-\eta|)=\sqrt{2\ell}=\Omega(\sqrt{\ell})$, as desired.

Over larger alphabets the expectation of $\delta$ can only increase. The easiest way to see this is to split $\Sigma$ arbitrarily into two subsets $K_1$ and $K_2$ of equal size. Then $x$ with respect to $z$ has a deficiency of letters from one of these subsets, say, $K_1$. By the argument for the binary alphabet, $\sqrt{2\ell}$ additional letters is needed, on expectation, to cover this deficiency. This is a necessary (but not sufficient) condition to obtain the word $v$. Hence $E(\delta)=\Omega(\sqrt{\ell})$. \qed
\end{proof}

Algorithm~\ref{alg:asmall} significantly speeds up the naive algorithm but is still rather slow. For the case $\alpha\le 3/2$ a much faster dictionary-based Algorithm~\ref{alg:dict} is presented below. However, Algorithm~\ref{alg:dict} consumes the amount of memory which is quadratic in $n$; this limits the depth of the search to the values of about $10^4$ for an ordinary laptop.

When processing a word $u=u[1..n]$, the algorithm has in the dictionary all factors of $u[1..n{-}1]$ up to the Abelian equivalence. Recall that a dictionary contains a set of pairs (key, value), where all keys are unique, and supports fast lookup, addition and deletion by key. For the dictionary used in Algorithm~\ref{alg:dict}, the keys are Parikh vectors and the values are lists of positions, in the increasing order, of the factors having this Parikh vector. The algorithm accesses only the last (maximal) element of the list. Let us describe the updates of the dictionary (lines 3 and 8 of Algorithm~\ref{alg:dfs}). At line 3, we delete all suffixes of $u$ from the dictionary. For a suffix $z$, this means the deletion of the last element from the list $dict[\Psi(z)]$; if the list becomes empty, the entry for $\Psi(z)$ is also deleted. At line 8, all suffixes of $ua$ are added to the dictionary. For a suffix $z$, if $\Psi(z)$ was not in the dictionary, an entry is created; then the position $|ua|-|z|+1$ is added to the end of the list $dict[\Psi(z)]$.

\begin{algorithm*}[!htb]
\caption{Dictionary-based Abelian powers detection ($\alpha\le 3/2$)} 
\label{alg:dict}
\begin{algorithmic}[1]
\State{\textbf{function} $\mathsf{alphafreedict}(u)$ \Comment{$u$=word; $n=|u|$}}
\State{$\free\gets \mathsf{true}$ \Comment{$\alpha$-A-freeness flag}}
\State{$\vec P\gets \vec 0$}
\For{$i=n$ downto 1+$\lceil n/2\rceil$} \Comment{$z=u[i..n]$}
    \State{$len\gets n-i+1$ \Comment{length of $z$}}
    \State{$\vec P[u[i]] \gets P[u[i]]+1$ \Comment{get $\Psi(z)$ from $\Psi(u[i{+}1..n])$}}
    \State{$pos\gets dict[\vec P].last$\Comment{position of last occurrence of some $x\sim z$, if exists}}
    \If{$pos \le i - len$ \textbf{and} $pos\ge \lceil\frac{\alpha i-1-n}{\alpha-1}\rceil$} \Comment{$xyz$ is forbidden}
       \State{$\free\gets \mathsf{false}$; break}
    \EndIf
\EndFor
\State {return $\free$ \Comment{the answer to ``is $u$ $\alpha$-A-free?''}}
\end{algorithmic}
\end{algorithm*}

\begin{proposition} \label{p:corr-adict}
Let $\alpha$ be a number such that $1<\alpha\le 3/2$ and $u$ be a word all proper prefixes of which are $\alpha$-A-free. Then Algorithm~\ref{alg:dict} correctly detects whether $u$ is $\alpha$-A-free.
\end{proposition}
\begin{proof}
Let us show that Algorithm~\ref{alg:dict} verifies $(\star)$. First suppose that the algorithm returned $\mathsf{false}$. Then it broke from the \textbf{for} cycle (line 9); let $z$ be the last suffix processed. The lookup by the key $\Psi(z)$ returned the  position of a factor $x\sim z$, and the condition in line 8 was true. Then the suffix $xyz=u[pos..n]$ violates $(\star)$. Indeed, $x\sim z$ since $x$ was found by the key $\Psi(z)$; the first inequality means that $x$ and $z$ do not overlap in $u$; the second inequality is equivalent to $\frac{|xyz|}{|xy|}\ge\alpha$. 

Now suppose that the algorithm returned $\mathsf{true}$. Aiming at a contradiction, assume that $u$ has a suffix violating $(\star)$. Let $xyz$ $(x\sim z)$ be the shortest such suffix. Consider the iteration of the \textbf{for} cycle, where $z$ was processed. The key $\Psi(z)$ was present in the dictionary because $x\sim z$. If $pos$ (line 7) corresponded to the $x$ from our ``bad'' suffix, i.e., $xyz=u[pos..n]$, then both inequalities in line 8 held because $x$ and $z$ do not overlap in $u$ and $\frac{|xyz|}{|xy|}\ge\alpha$. But then the algorithm would have returned $\mathsf{false}$, contradicting our assumption. Hence $pos$ was the position of some other $x'\sim z$ which occurs in $u$ later than $x$. By the choice of the suffix $xyz$, $u$ cannot have shorter suffix $x'y'z$ with $x'\sim z$. This means that the occurrences of $x'$ and $z$ overlap (see Fig.~\ref{f:closest}). 

\begin{figure}[!htb]
\centering
    \includegraphics[trim = 50 772 330 38, clip]{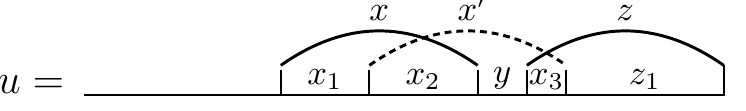}
    \caption{Location of Abelian equivalent factors (Proposition~\ref{p:corr-adict}).}
\label{f:closest}
\end{figure}

Note that $x'$ and $x$ also overlap. Otherwise, $xyz$ has a prefix of the form $x\hat yx'$ and $\frac{|x\hat yx'|}{|x'\hat y|}\ge \frac{|xyz|}{|xy|}\ge \alpha$, contradicting the condition that all proper prefixes of $u$ are $\alpha$-A-free. Then $x=x_1x_2$, $x'=x_2yx_3$, $z=x_3z_1$, as shown in Fig.~\ref{f:closest}, and $x_1,x_2,x_3,z_1$ are nonempty. We observe that $x\sim x'\sim z$ imply $x_1\sim yx_3$ and $x_2y\sim z_1$. By the condition on the prefixes of $u$, $\frac{|x_1x_2yx_3|}{|x_1x_2|}<\alpha\le 3/2$. Hence $|yx_3|<|x_2|$ and then $|x_3|<|x_2y|=|z_1|$. Therefore $\frac{|x_2yx_3z_1|}{|x_2yx_3|}> 3/2\ge \alpha$, so the suffix $x'z_1=x_2yx_3z_1$ of $u$ violates $(\star)$. But $|x'z_1|<|xyz|$, contradicting the choice of $xyz$. This contradiction proves that $u$ satisfies $(\star)$. \qed
\end{proof}

Dictionaries based on hash table techniques such as linear probing or cuckoo hashing guarantee expected constant time per dictionary operation. As Algorithm~\ref{alg:dict} consists of a single cycle with a constant number of operations inside, the following statement is straightforward.

\begin{proposition} \label{p:dictcompl}
For a word of length $n$, Algorithm~\ref{alg:dict} performs $O(n)$ operations, including dictionary operations.
\end{proposition}

\begin{remark} \label{r:dict32}
A slight modification of Algorithm~\ref{alg:dict} allows one to process the important case $\alpha=(3/2)^+$ within the same complexity bound. The whole argument from the proof of Proposition~\ref{p:corr-adict} remains valid for $\alpha=(3/2)^+$ except for one specific situation: in Fig.~\ref{f:closest} it is possible that $y$ is empty and $|x_1|=|x_2|=|x_3|=|x_4|$. In this situation Algorithm~\ref{alg:dict} misses the Abelian square $xz$. To fix this, we add a patch after line 7:\\[1mm]
7.5: \textbf{if} $pos= i- len/2$ \textbf{then} $pos \gets pos.next$\\[1mm]
As an example, consider $u=abcdbadc$. Processing the suffix $z=badc$, Algorithm~\ref{alg:dict} retrieves $pos=3$ from the dictionary by the key $\Psi(z)$. The corresponding factor $x'=cdba$ overlaps $z$ and the condition in line 8 would fail for $pos$. However, $pos$ satisfies the condition in the inserted line 7.5 and thus the factor $x=abcd$ at $pos=1$ will be reached. The condition in line 8 holds for $pos=1$ and the forbidden Abelian repetition is detected.
\end{remark}
\begin{remark} \label{r:dict2}
Algorithm~\ref{alg:dict} can be further modified to work for all $\alpha<2$. If we replace the patch from Remark~\ref{r:dict32} with the following one:\\[1mm]
7.5: \textbf{while} $pos> i- len$ \textbf{do} $pos \gets pos.next$\\[1mm]
the algorithm will find the closest factor $x\sim z$ which does not overlap with $z$. This new patch introduces an inner cycle and thus affects the time complexity but the algorithm remains faster in practice than Algorithm~\ref{alg:asmall}.
\end{remark}

\subsection{Avoiding big powers} \label{ss:big}
Let $\alpha>2$. The case $2<\alpha<3$ for ternary words and the case $3<\alpha<4$ for binary words are relevant to the studies of Abelian repetition threshold. We provide here the algorithms for the first case; the algorithms for the second case are very similar (the only difference is that one should check for Abelian cubes instead of Abelian squares). An Abelian $\beta$-power with $\alpha\le \beta\le 3$ has the form $ZZ'z$, where $Z\sim Z'$,  $z$ is equivalent to a prefix of $Z$, and $\frac{|ZZ'z|}{|Z|}\ge \alpha$. We can write the Abelian square $ZZ'$ as $xy$ where $|x|\le|y|$ and $x\sim z$. Consequently, if all proper prefixes of a word $u$ are $\alpha$-A-free, then $u$ is $\alpha$-A-free iff the following analog of $(\star)$ holds:
\begin{itemize}
    \item[$(*)$] no suffix of $u$ can be written as $xyz$ such that $|y|\ge |x|>0$, $x\sim z$, $xy$ is an Abelian square, and $\frac{2|xyz|}{|xy|}\ge\alpha$.
\end{itemize}
Verifying $(*)$ for $u$, we proceed for its suffix $z$ as follows. Within the range determined by $\alpha$, we search for all factors $x=u[\lef..right]$ such that $x\sim z$. The search is organized as in Algorithm~\ref{alg:asmall} (see Fig.~\ref{f:jumping}). For each $x$ we consider the corresponding suffix $xyz$ of $u$ and check whether $xy$ is an Abelian square. If yes, $xyz$ violates $(*)$. In Algorithm~\ref{alg:abig} below, we make use of the arrays $c_a,d_a$ and functions $\pari,\cov$ designed for Algorithm~\ref{alg:asmall}.

\begin{algorithm*}[!htb]
\caption{Abelian powers detection (case $2<\alpha<3$)} 
\label{alg:abig}
\begin{algorithmic}[1]
\State{\textbf{function} $\mathsf{ALPHAfree}(u)$ \Comment{$u$=word; $n=|u|$}}
\State{$\free\gets \mathsf{true}$ \Comment{$\alpha$-A-freeness flag}}
\For{$i=n$ downto 1+$\lceil 2n/3\rceil$} \Comment{$z=u[i..n]$}
    \State{$right\gets i-1$}
    \State{$len\gets n-i+1$; $\vec P\gets \pari(i,n)$ \Comment{length and Parikh vector of $z$}}
    \State{$\lef\gets \cov(\vec P, right)$}
    \If{$\lef=0$} break \Comment{$\Psi(u[1..right])\not\ge \Psi(z)$}
    \EndIf
    \While {$\lef\ge \max\{1, \lceil\frac{\alpha i-1-2n}{\alpha-2}\rceil\}$}\Comment{guarantees $\frac{2|xyz|}{|xy|}\ge\alpha$}
        \If{$\lef+len-1=right$} \Comment{$x=u[\lef..right]\sim z$}
            \If{$2\mid(i-\lef)$ \textbf{and} $\sum_{a\in\Sigma} \big|c_a[i{-}1]+c_a[\lef{-}1]-2c_a[\frac{i+\lef}{2}{-}1]\big| = 0$}
                \State{$\free\gets \mathsf{false}$; break}
                \Comment{$xy=u[\lef..i{-}1]$ is an Abelian square}
            \Else \State{$right\gets right -1$}\Comment{right bound for the next search}
            \EndIf
        \Else \State{$right\gets \lef+len-1$}\Comment{right bound for the next search}
        \EndIf
        \State{$\lef\gets \cov(\vec P, right)$}
    \EndWhile
    \If{$\free=\textsf{false}$} break
    \EndIf
\EndFor
\State {return $\free$ \Comment{the answer to ``is $u$ $\alpha$-A-free?''}}\end{algorithmic}
\end{algorithm*}

\begin{proposition} \label{p:corr-abig}
Let $\alpha$ be a number such that $2<\alpha<3$ and $u$ be a word all proper prefixes of which are $\alpha$-A-free. Then Algorithm~\ref{alg:abig} correctly detects whether $u$ is $\alpha$-A-free.
\end{proposition}
\begin{proof}
Algorithm~\ref{alg:abig} is similar to Algorithm~\ref{alg:asmall}, so we focus on their difference. If some suffix $xyz$ violates $(*)$, then $|z|\le|xyz|/3\le n/3$; hence the range for the outer cycle in line 3. For a fixed $z$ we repeatedly seek for the shortest factor $v=u[\lef..right]$ with the given right bound and the property $\Psi(v)\ge \Psi(z)$. If $|v|=|z|$ (condition in line 9 holds), then $v$ is a candidate for $x$ in the suffix $xyz$ violating $(*)$. The initial value for $right$ (line 4) is set to ensure the condition $|x|\le|y|$. The candidate found in line 9 is  checked in line 10 for the remaining condition: $xy$ is an Abelian square. Namely, we check that $|xy|$ is even and its left and right halves have the same Parikh vector. If this condition holds, the algorithm broke both inner and outer cycles and returns $\mathsf{false}$. If the condition fails, we decrease $right$ by 1 and compute the factor $v$ for this new right bound. The rest is the same as in Algorithm~\ref{alg:asmall}. So we can conclude that Algorithm~\ref{alg:abig} verifies $(*)$. \qed
\end{proof}
\begin{remark} \label{r:time3}
The inner cycle of Algorithm~\ref{alg:abig} works in $O(\sigma)$ time, and so Algorithm~\ref{alg:abig} has the same time complexity $O((K+n)\sigma)$ as
Algorithm~\ref{alg:asmall}. (Recall that $K$ is the total number of iterations of the inner cycle during processing $u$.)
\end{remark}

Algorithm~\ref{alg:abig} is rather slow. But it appears that \textit{dual} Abelian powers can be detected by a much faster Algorithm~\ref{alg:dual}. Let us give the definitions.
As was mentioned in Section~\ref{s:pre}, the reversal of an $\alpha$-A-power for $\alpha>2$ is not necessarily an $\alpha$-power. For example, $u=abc\,bac\,a$ is a $(7/3)$-A-power while $u^R=aca\,bcb\,a$ does not begin with an Abelian square. We call $u$ a \emph{dual $\alpha$-A-power} if $u^R$ is an $\alpha$-A-power; the notion of dual $\alpha$-A-free word is defined by analogy with $\alpha$-A-free word. Dual $\alpha$-A-free words are exactly the reversals of $\alpha$-A-free words. 

Assume that all proper prefixes of a word $u$ are dual $\alpha$-A-free, where $2< \alpha< 3$. Then $u$ is dual $\alpha$-A-free iff the following analog of $(*)$ holds:
\begin{itemize}
    \item[$(\dagger)$] no suffix of $u$ can be written as $xyz$ such that $y\sim z$, $x$ is equivalent to a suffix of $z$, and $\frac{|xyz|}{|z|}\ge\alpha$.
\end{itemize}


\begin{algorithm*}[!htb]
\caption{Dual Abelian powers detection ($2<\alpha< 3$)} 
\label{alg:dual}
\begin{algorithmic}[1]
\State{\textbf{function} $\mathsf{dualALPHAfree}(u)$ \Comment{$u$=word; $n=|u|$}}
\State{$\free\gets \mathsf{true}$ \Comment{$\alpha$-A-freeness flag}}
\State{$i\gets n$}
\While{$i\ge 1+\lceil \frac{\alpha-1}{\alpha}n \rceil$} \Comment{$z=u[i..n]$}
    \State{$len\gets n-i+1$; $\vec P\gets \pari(i,n)$ \Comment{length and Parikh vector of $z$}}
    \State{$\lef\gets \cov(\vec P, i-1)$}\Comment{computing $v$}
    \If{$\lef+len=i$} \Comment{$|v|=|z|\Rightarrow v=y\sim z$}
        \State{$j=\lceil (\alpha-2)\cdot len\rceil$}\Comment{minimal length of $x$}
        \While {$j\le len$}\Comment{possible lengths of $x$}
            \State{$\vec P_1\gets \pari(n{-}j{+}1,n)$ \Comment{Parikh vector of the length-$j$ suffix of $z$}}                \State{$\lef_1\gets \cov(\vec P_1, \lef-1)$}\Comment{computing $v_1$ for $x$}
            \If{$\lef_1+j=\lef$} \Comment{$x$ is found} 
                \State{$\free\gets \mathsf{false}$; break}
            \Else
                \State{$j\gets \lef-\lef_1$}
            \EndIf
        \EndWhile\vspace*{-2pt} 
        \State{$i\gets i-1$}
    \Else
        \State{$i\gets \lceil (n+\lef)/2 \rceil$}
    \EndIf
    \If{$\free=\textsf{false}$} break
    \EndIf
\EndWhile
\State {return $\free$ \Comment{the answer to ``is $u$ dual $\alpha$-A-free?''}}
\end{algorithmic}
\end{algorithm*}

\begin{proposition} \label{p:dual}
Let $\alpha$ be a number such that $2<\alpha<3$ and $u$ be a word all proper prefixes of which are dual $\alpha$-A-free. Then Algorithm~\ref{alg:dual} correctly detects whether $u$ is dual $\alpha$-A-free.
\end{proposition}
\begin{proof}
If some suffix $xyz$ violates $(\dagger)$, then $|z|\le|xyz|/\alpha\le n/\alpha$; hence the range for the outer cycle in line 3. The general scheme is as follows. For each processed suffix $z$, the algorithm first checks if $u$ ends with an Abelian square $yz$ ($y\sim z$); if yes, it checks whether $yz$ is preceded by some $x$ which is equivalent to a suffix of $z$. If such an $x$ is found, the algorithm detects a violation of $(\dagger)$ and stops. If either $x$ or $y$ is not found, the algorithm moves to the next appropriate suffix. Let us consider the details.

In line 6, the shortest $v=u[\lef..i{-}1]$ such that $vz$ is a suffix of $u$ and $\Psi(v)\ge \Psi(z)$ is computed. If $|v|=|z|$ (the condition in line 7), then $y=v$ is found and we enter the inner cycle to find $x$. If $|v|>|x|$, we note that the suffixes of $u$ of lengths between $2|z|$ and $|vz|-1$ cannot be Abelian squares; then the next suffix to be considered has the length $\lceil \frac{|vz|}{2}\rceil$, as is set in line 18. In the inner cycle, a similar idea is implemented: for each processed suffix $z_1$ of $z$ the algorithm finds the shortest word $v_1=u[\lef_1..\lef-1]$ satisfying $\Psi(v_1)\ge \Psi(z_1)$ (line 11). If $|v_1|=|z_1|$ (line 12), then $x$ is found; otherwise, the next suffix of $z$ to be checked is of length $|v_1|$ (line 15). The inner cycle breaks if this length exceeds the length of $z$. 

We have shown that Algorithm~\ref{alg:dual} stops with the answer $\mathsf{false}$ if it finds a suffix $xyz$ of $u$ that violates $(\dagger)$; if it finishes the check of the suffix $z$ without breaking or skips this suffix at all, then $u$ has no suffix $xyz$, violating $(\dagger)$. Therefore, the algorithm verifies $(\dagger)$. \qed
\end{proof}

Algorithm~\ref{alg:dual} works extremely fast compared to other algorithms from this section. The following statement holds for the case of the ternary alphabet.
\begin{proposition} \label{p:dualcompl}
For a word $u$ picked up uniformly at random from the set $\{0,1,2\}^n$, Algorithm~\ref{alg:dual} works in $\Theta(\sqrt{n})$ expected time.
\end{proposition}
\begin{proof}
Lemma~\ref{l:prob} says that the expected length of the word $v$ found in line 6 is $|z|+\Omega(\sqrt{|z|})$ and thus, on expectation, the assignment in line 18 leads to skipping $\Omega(\sqrt{|z|})$ suffixes of $u$. Hence the expected total number of processed suffixes of $u$ is $O(\sqrt{n})$. By the same lemma, the inner cycle for a suffix $z$ runs, on expectation, $O(\sqrt{|z|})$ iterations, so its expected time complexity is  $O(\sqrt{|z|})$. Thus, processing the suffix of length $\ell$, Algorithm~\ref{alg:dual} performs $O(1)+p_\ell\cdot O(\sqrt{\ell})$ operations, where $p_\ell$ is the probability to enter the inner cycle, i.e., the probability that two random ternary words of length $\ell$ are Abelian equivalent. Let us estimate this probability. One has 
$$
p_\ell \le \max_{k_1,k_2,k_3}\binom{\ell}{k_1,k_2,k_3}
\Big/ 3^\ell,
$$
where the denominator is the number of ternary words of length $\ell$ and the numerator is the maximum size of a class of Abelian equivalent ternary words of length $\ell$. This maximum, reached for (almost) equal $k_1,k_2,k_3$, can be estimated by the Stirling formula to $\Theta(3^\ell/\ell)$. Thus $p_\ell=O(1/\ell)$. Then Algorithm~\ref{alg:dual} performs, on expectation, $O(1)$ operations per iteration of the outer cycle. The result now follows. \qed
\end{proof}

\section{Experimental results} \label{s:exhaust}

We ran a big series of experiments for $\alpha$-A-free languages over the alphabets of size $2,3,\ldots,10$. Each of the experiments is a set of random walks in the prefix tree of a given language. Each walk follows the random depth-first search (Algorithm~\ref{alg:dfs}), with the number $N$ of visited nodes being  of order $10^5$ to $10^7$. The ultimate aim of every experiment was to make a well-grounded conjecture about the (in)finiteness of the studied language.

Our initial expectation was that random walks will demonstrate two easily distinguishable types of behaviour:
\begin{itemize}
    \item \emph{infinite-like}: the level of the current node is (almost) proportional to the number of nodes visited, or
    \item \emph{finite-like}: from some point, the level of the current node oscillates near the maximum reached earlier.
\end{itemize}
However, the situation is not that straightforward: very long oscillations of level were detected during random walks even in some languages which are \emph{known} to be infinite; for example, in the binary 4-A-free language. To overcome such an unwanted behaviour, we endowed Algorithm~\ref{alg:dfs} with a ``forced backtrack'' rule:
\begin{itemize}
    \item let $\ml=k$ be the maximum level of a node reached so far; if $f(k)$ nodes were visited since the last update of $\ml$ or since the last forced backtrack, then make a forced backtrack: from the current node, move $g(k)$ edges up the tree and continue the search from the node reached.
\end{itemize}
Here $f(k)$ and $g(k)$ are some heuristically chosen monotone functions; we used $f(k)=\lceil k^{3/2}\rceil$ and $g(k)=\lceil k^{1/2}\rceil$. Forced backtracking deletes the last $g(k)$ letters of the current word in order to get out of a big finite subtree the search was supposed to traverse. The use of forced backtracking allowed us to classify the walks in almost all studied languages either as infinite-like or as finite-like. The results presented below are grouped by the alphabets.

\subsection{Alphabets with $6,7,8,9$, and 10 letters}

In \cite{SaSh12}, it was proved (Theorem~3.1) that $\ART(k)\ge \frac{k-2}{k-3}$ for all $k\ge 5$ and conjectured that the equality holds in all cases. However, the random search reveals a different picture. For each of the languages $\AF(k,\frac{k-2}{k-3}^+)$, $k=6,7,8,9,10$, we ran random search with forced backtrack, using Algorithm~\ref{alg:dict} to decide the membership in the language; the search terminated when $N$ nodes were visited. We repeated the search 100 times with with $N=10^6$ and another 100 times with $N=2\cdot 10^6$. The results, presented in columns 3--8 of Table~\ref{tab:678910}, clearly demonstrate finite-like behaviour of random walks. Moreover, the results suggest that neither of these languages contains a word much longer than 100 symbols. So we were able to prove the following result by (optimized) exhaustive search.

\begin{theorem} \label{t:678910}
One has $\ART(k)>\frac{k-2}{k-3}$ for $k=6,7,8,9,10$.
\end{theorem}

\begin{table}[!bht]
    \centering
    \begin{tabular}{|c|c|c|c|c|c|c|c|c|}
    \hline
    \multirow{2}{1.4cm}{Alphabet size} & \multirow{2}{1.2cm}{Avoided power} & 
    \multicolumn{3}{c|}{$N=10^6$} & \multicolumn{3}{c|}{$N=2\cdot 10^6$}& 
    \multirow{2}{1.5cm}{Maximum length}\\
    \cline{3-8}
    && $\ml_{max}$ & $\ml_{av}$ & $\ml_{med}$  
    & $\ml_{max}$ & $\ml_{av}$ & $\ml_{med}$& \\
    \hline
    6 & $(4/3)^+$& 112&98.9&98& 114&101.1&101&116\\
    7 & $(5/4)^+$& 116&100.3&100& 124&103.9&102&125\\
    8 & $(6/5)^+$& 103&94.8&95& 102&96.2&96&105\\
    9 & $(7/6)^+$& 108&95.6&96& 107&98.8&99&117\\
    10 & $(8/7)^+$& 121&107.7&108& 128&111.6&111&148*\\
    \hline 
    \end{tabular}
    \vspace*{2mm}\caption{Maximum levels $\ml$ reached by random walks in some Abelian power-free languages. Columns 3--5 (resp. 6--8) show the maximum, average, and median values of $\ml$ among 100 random walks visiting $N=10^6$ (resp., $N=2\cdot 10^6$) nodes each. Column~9 shows the length of a longest word in the language, found by exhaustive search.} 
    \label{tab:678910}
\end{table}

A length-$n$ word is called $n$-permutation if all its letters are pairwise distinct. We use the following lemma to reduce the search space.

\begin{lemma} \label{l:search}
Let $k\ge 6$, $\alpha={\frac{k-2}{k-3}}^+$, and let $L_1,L_2$, and $L_3$ be subsets of $L=\AF(k,\alpha)$ defined as follows:\\
- $L_1$ is the set of all $w\in L$ such that $w$ has the prefix $01\cdots(k{-}3)$ and contains no $(k-1)$-permutations;\\
- $L_2$ is the set of all $w\in L$ such that $w$ has the prefix $01\cdots(k{-}2)$ and contains no $k$-permutations;\\
- $L_3$ is the set of all $w\in L$ having the prefix $01\cdots(k{-}1)$.\\
Then $L$ is finite iff each of $L_1,L_2$, and $L_3$ is finite.
\end{lemma}
\begin{proof}
The necessity is immediate from definitions; let us prove sufficiency. Let $w\in L$ and let $\ell_1,\ell_2,\ell_3$ be the lengths of the longest words in $L_1,L_2$, and $L_3$ respectively.
Let us show that $|w|<\ell_1+\ell_2+\ell_3$. If $u$ is a word and $a$ is a letter, then $aua$ is an Abelian $\frac{|u|+2}{|u|+1}$-power. Hence 
\begin{itemize}
    \item[$(\diamond)$] any factor of $w$  of length $k-3$ is a $(k-3)$-permutation.
\end{itemize} 
Now consider a factor $u$ of $w$ such that $|u|=k-1$. By $(\diamond)$, one can write $u=abu'cd$, where $u'$ does not contain the letters $a,b,c,d$ and $b\ne c$. If $a=c$ and $b=d$, then $u$ is an Abelian $\frac{k-1}{k-3}$-power, which is impossible since $u\in L$ as a factor of $w$. Hence $u$ either begins or ends with a $(k-2)$-permutation. Thus $w$ contains a $(k-2)$-permutation beginning at the first or second position. Then $w$ contains a $(k-1)$-permutation due to finiteness of $L_1$; moreover, this permutation ends no later than at position $2+\ell_1$ and thus begins no later than at position $4-k+\ell_1$. Similarly, due to finiteness of $L_2$, $w$ contains a $k$-permutation no later than at position $5-2k+\ell_1+\ell_2$. Finally, the finiteness of $L_3$ implies the upper bound $|w|\le 4-2k+\ell_1+\ell_2+\ell_3$. In particular, $L$ is finite.\qed
\end{proof}

We ran (non-randomized) depth-first search on the prefix trees of the languages $L_1,L_2$, and $L_3$ for the cases $k=6,7,8,9,10$, using Algorithm~\ref{alg:dict} to detect Abelian powers, and proved that all these trees are finite. According to Lemma~\ref{l:search}, this proves Theorem~\ref{t:678910}. The total number of visited nodes was approximately 0.43 billions for $k=8$; 0.90 billions for $k=7$; 6.29 billions for $k=6$; 8.14 billions for $k=9$; more than 500 billions for $k=10$. The last case required about 2000 hours of processing time (single-core) by an ordinary laptop.

\begin{remark} \label{r:exhaust}
For each $k=6,7,8,9$ it is feasible to run a single search which enumerates all lexmin words in the language $\AF(k,\frac{k-2}{k-3}^+)$. We performed these searches and found the maximum length of a word in each language (the last column of Table~\ref{tab:678910}) and the distribution of words by their length (see Fig.~\ref{f:8ary} for the case $k=8$). For $k=10$, such a single search would require too much resources; here the value in the last column of Table~\ref{tab:678910} is the length of the longest word in $L_1\cup L_2\cup L_3$.
\end{remark}

\begin{figure}[!htb]
\centering
    \includegraphics[scale=0.69]{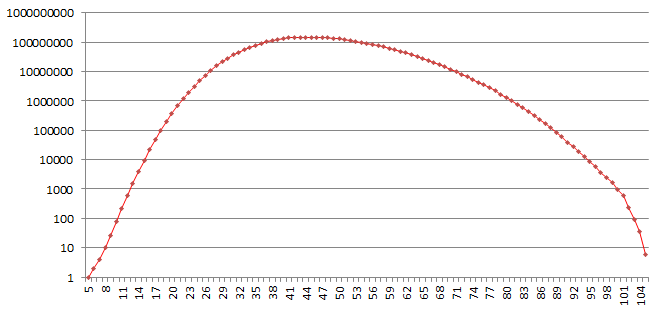}
    \caption{Distribution of the number of lexmin words by length in the language $\AF(8,\frac{6}{5}^+)$ (logarithmic scale).}
    \label{f:8ary}
\end{figure}

Theorem~\ref{t:678910} raises the question of avoidance of bigger $\alpha$-A-powers over the same alphabets. As the next step, we ran experiments for the languages $\AF(k,\frac{k-3}{k-4})$. The results for $k=7,8,9,10$ are presented in Table~\ref{tab:678910_}; random walks in these languages clearly demonstrate finite-type behaviour, while proving finiteness by exhaustive search looks hardly possible. On the contrary, the walks in the 6-ary language $\AF(6,\frac{3}{2})$ demonstrate an infinite-like behaviour: the average value of $\ml$ for our experiments with $N=10^5$ is greater than $5\cdot 10^4$. We note that the obtained words are too long to use Algorithm~\ref{alg:dict}, so we had to use slower Algorithm~\ref{alg:asmall}. Finally, we constructed random walks for the languages $\AF(k,\frac{k-3}{k-4}^+)$ ($k=7,8,9,10$). They also demonstrate infinite-like behaviour. The obtained experimental results allow us to state the part of Conjecture~\ref{c:art} for the alphabets with 6 and more letters.

\begin{table}[!htb]
    \centering
    \begin{tabular}{|c|c|c|c|c|c|c|c|}
    \hline
    \multirow{2}{1.4cm}{Alphabet size} & \multirow{2}{1.2cm}{Avoided power} & 
    \multicolumn{3}{c|}{$N=10^6$} & \multicolumn{3}{c|}{$N=2\cdot 10^6$}\\
    \cline{3-8}
    && $\ml_{max}$ & $\ml_{av}$ & $\ml_{med}$  
    & $\ml_{max}$ & $\ml_{av}$ & $\ml_{med}$ \\
    \hline
    7 & $4/3$& 510&374.5&371& 510&397.5&394\\
    8 & $5/4$& 211&179.7&179& 223&185.0&184\\
    9 & $6/5$& 192&157.2&156& 191&162.3&161\\
    10 & $7/6$& 175&154.0&154& 187&159.7&158\\
    \hline 
    \end{tabular}
    \vspace*{2mm}\caption{Maximum levels $\ml$ reached by random walks in some Abelian power-free languages. Columns 3--5 (resp. 6--8) show the maximum, average, and median values of $\ml$ among 100 random walks visiting $N=10^6$ (resp., $N=2\cdot 10^6$) nodes each.}
    \label{tab:678910_}
\end{table}

\subsection{Alphabets with $2,3,4$, and 5 letters}

Random walks in the prefix tree of the language $\AF(5,\frac{3}{2}^+)$ demonstrate the infinite-like behaviour; Fig.~\ref{f:5ary} shows an example of dependence of the level of the current node on the number of nodes visited. Although we could not push random walks significantly farther down the tree (Algorithm~\ref{alg:dict} uses too much space to work on such big levels, so we relied on slow Algorithm~\ref{alg:asmall}), we obtained sufficient evidence to support Conjecture~\ref{c:artold} for $k=5$.

\begin{figure}[!htb]
\centering
    \includegraphics[scale=0.8]{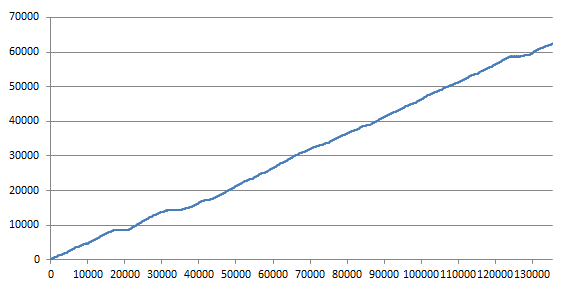}
    \caption{An infinite-like random walk in the language $\AF(5,\frac{3}{2}^+)$: a point $(n,m)$ of the graph means that the $n$th node visited by the walk has depth $m$.}
    \label{f:5ary}
\end{figure}

\begin{remark} \label{r:5arygrowth}
As the language $\AF(5,\frac{3}{2}^+)$ is supposed to be infinite, it is interesting to estimate its growth. Based on the technique described in \cite{PeSh21dlt}, we estimate the number of words of length $n$ in $\AF(5,\frac{3}{2}^+)$ as growing exponentially with $n$ at the rate close to 1.5. The upper bound 2.335 \cite{SaSh12} on this rate is thus very loose.
\end{remark}

For smaller alphabets, we studied the languages $\AF(4,\frac{9}{5}^+)$, $\AF(3,2^+)$, and $\AF(2,\frac{11}{3}^+)$, indicated by Conjecture~\ref{c:artold} as infinite. To detect Abelian powers, we used Algorithm~\ref{alg:dict} for the quaternary alphabet; for the ternary and binary alphabets, we worked with the reversals of $\AF(3,2^+)$, and $\AF(2,\frac{11}{3}^+)$ to benefit from the speed of Algorithm~\ref{alg:dual} detecting \emph{dual} Abelian powers. Random walks (with forced backtracking) in each of three languages show the finite-like behaviour; see Table~\ref{tab:234} and the example in Fig.~\ref{f:4ary}. Longer random searches lead to somewhat better results, especially on average, due to multiple forced backtracks; however, nothing resembles a steady growth of maximum level with the total number of visited nodes. So the experimental results justify lower bounds from Conjecture~\ref{c:art} for $k=2,3,4$. To get the upper bound for the ternary alphabet, we ran random walks for the language $\AF(3,\frac{5}{2}^+)$ with the results similar to those obtained for $\AF(5,\frac{3}{2}^+)$: all walks demonstrate the infinite-like behaviour; the level  $\ml=10^5$ is reached within minutes.

Overall, we conclude that the experiments we conducted justify the formulation of Conjecture~\ref{c:art}.


\begin{table}[!htb]
    \tabcolsep=1.8pt
    \centering
    \begin{tabular}{|c|c|c|c|c|c|c|c|c|c|c|}
    \hline
    \multirow{2}{1.32cm}{Alphabet size} & \multirow{2}{1.2cm}{Avoided power} & 
    \multicolumn{3}{c|}{$N=10^6$} & \multicolumn{3}{c|}{$N=2\cdot 10^6$} &
    \multicolumn{3}{c|}{$N=10^7$}\\
    \cline{3-11}
    && $\ml_{max}$ & $\ml_{av}$ & $\ml_{med}$  
    & $\ml_{max}$ & $\ml_{av}$ & $\ml_{med}$  
    & $\ml_{max}$ & $\ml_{av}$ & $\ml_{med}$\\
    \hline
    2 & $(11/3)^+$& 775&435.8&416& 706&477.0&453 &759&589.7&588 \\
    3 & $2^+$& 3344&1700.0&1671& 5363&2228.8&2140& 5449&3078.1&3148\\
    4 & $(9/5)^+$& 1367&861.2&835& 1734&986.8&956& 2453&1414.7&1369\\
    \hline 
    \end{tabular}
    \vspace*{2mm}\caption{Maximum levels $\ml$ reached by random walks in some Abelian power-free languages. Columns 3--5 (resp. 6--8, 9--11) show the maximum, average, and median values of $\ml$ among 100 random walks visiting $N=10^6$ (resp., $N=2\cdot 10^6$, $N=10^7$) nodes each.}
    \label{tab:234}
\end{table}

We ran two additional experiments with the language $\AF(4,\frac{9}{5}^+)$. First, we took the longest word found by random search (it has length 2453) and extended it to the left with another long random search, repeated multiple times. The longest obtained word has length 3152, which seems to be a fair approximation of the maximum length of a word in $\AF(4,\frac{9}{5}^+)$. Second, we tried an exhaustive enumeration of the words in $\AF(4,\frac{9}{5}^+)$ to understand how fast is the initial growth and how far we can reach. We discovered that the language contains $10.68$ billions of lexmin words of length 90 compared to $9.49$ billions of such words of length 89. Hence, 90 is still quite far from the length where the number of words will be maximal.

\begin{figure}[!htb]
\centering
    \includegraphics[scale=0.65]{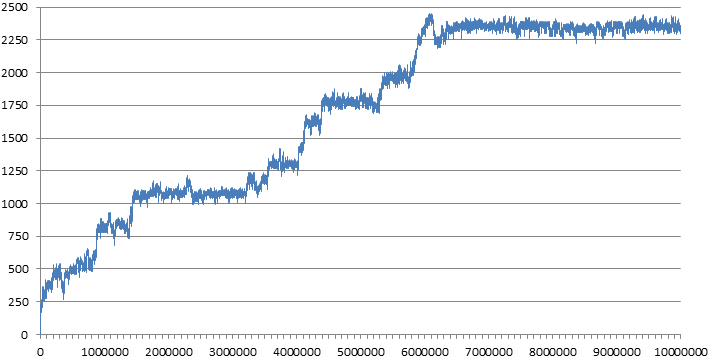}
    \caption{A finite-like random walk in the language $\AF(4,\frac{9}{5}^+)$: a point $(n,m)$ of the graph means that the $n$th node visited by the walk has depth $m$.}
    \label{f:4ary}
\end{figure}



\section{Future Work} \label{s:fin}

Clearly, the main challenge in the topic is to find the exact values of the Abelian repetition threshold. Even finding one such value would be a great progress. Choosing the case to start with, we would suggest proving $\ART(5)=3/2$ because in this case the lower bound is already checked by exhaustive search in \cite{SaSh12}. For all other alphabets, the proof of lower bounds suggested in Conjecture~\ref{c:art} is already a challenging task which cannot be solved by brute force.

Another piece of work is to refine Conjecture~\ref{c:art} by suggesting the precise values of $\ART(2)$, $\ART(3)$, $\ART(4)$, and $\ART(6)$. For bigger $k$, random walks demonstrate an obvious ``phase transition'' at the point $\frac{k-3}{k-4}$: the behaviour of a walk switches from finite-like for $\AF(k,\frac{k-3}{k-4})$ to infinite-like for $\AF(k,\frac{k-3}{k-4}^+)$. However, the situation with small alphabets can be trickier. For example, we tried the next natural candidate for $\ART(4)$, namely, $11/6$. For the random walks in $\AF(4,\frac{11}{6}^+)$, with $N=10^6$ and forced backtracks, the range of obtained maximum levels in our experiments varied from 3000 to 20000; such big lengths show that there is no hope to see a clear-cut phase transition in the experiments with random walks.

Finally, we want to draw attention to the following fact. The quaternary 2-A-free word constructed by Ker\"anen \cite{Ker92} contains arbitrarily long factors of the form $xa\bar x$, where $a$ is a letter and $x\sim \bar x$; thus it is not $\alpha$-A-free for any $\alpha<2$. Similarly, the word constructed by Dekking \cite{Dek79} for the ternary (resp., binary) alphabet is not $\alpha$-A-free for any $\alpha<3$ (resp., $\alpha<4$). Hence some new constructions are necessary to improve upper bounds for $\ART$.

\bibliographystyle{splncs04}
\bibliography{bibl}

\begin{thebibliography}{10}
\providecommand{\url}[1]{\texttt{#1}}
\providecommand{\urlprefix}{URL }
\providecommand{\doi}[1]{https://doi.org/#1}

\bibitem{Car07}
Carpi, A.: On {Dejean's} conjecture over large alphabets. Theoret. Comput. Sci.
   \textbf{385},  137--151 (1999)

\bibitem{CaCu99}
Cassaigne, J., Currie, J.D.: Words strongly avoiding fractional powers. Eur. J.
  Comb.  \textbf{20}(8),  725--737 (1999)

\bibitem{CuRa11}
Currie, J.D., Rampersad, N.: A proof of {Dejean's} conjecture. Math. Comp.
  \textbf{80},  1063--1070 (2011)

\bibitem{Dej72}
{Dejean}, F.: Sur un {th\'eor\`eme} de {Thue}. J. Combin. Theory. Ser. A
  \textbf{13},  90--99 (1972)

\bibitem{Dek79}
Dekking, F.M.: Strongly non-repetitive sequences and progression-free sets. J.
  Combin. Theory. Ser. A  \textbf{27},  181--185 (1979)

\bibitem{Erd61}
Erd\"os, P.: Some unsolved problems. Magyar Tud. Akad. Mat. Kutat\'o Int.
  K\"ozl.  \textbf{6},  221–264 (1961)

\bibitem{Evd68}
Evdokimov, A.A.: Strongly asymmetric sequences generated by a finite number of
  symbols. Soviet Math. Dokl.  \textbf{9},  536--539 (1968)

\bibitem{Ker92}
{Ker{\"a}nen}, V.: Abelian squares are avoidable on 4 letters. In: Kuich, W.
  (ed.) Proc. ICALP 1992. LNCS, vol.~623, pp. 41--52. Springer-Verlag (1992)

\bibitem{Mou92}
Moulin-Ollagnier, J.: Proof of {Dejean's} conjecture for alphabets with
  $5,6,7,8,9,10$ and $11$ letters. Theoret. Comput. Sci.  \textbf{95},
  187--205 (1992)

\bibitem{Pan84c}
Pansiot, J.J.: A propos d'une conjecture de {F. Dejean} sur les
  {r\'ep\'etitions} dans les mots. Discr. Appl. Math.  \textbf{7},  297--311
  (1984)

\bibitem{PeSh21dlt}
Petrova, E.A., Shur, A.M.: Branching frequency and {Markov} entropy of
  repetition-free languages. In: Developments in Language Theory - 25th
  International Conference, {DLT}, Proceedings. Lecture Notes in Computer
  Science, vol. 12811, pp. 328--341. Springer (2021)

\bibitem{Rao11}
Rao, M.: Last cases of {Dejean's} conjecture. Theoret. Comput. Sci.
  \textbf{412},  3010--3018 (2011)

\bibitem{SaSh12}
Samsonov, A.V., Shur, A.M.: On {Abelian} repetition threshold. RAIRO Theor.
  Inf. Appl.  \textbf{46},  147--163 (2012)

\bibitem{Thue06}
Thue, A.: {{\"U}ber} unendliche {Zeichenreihen}. Norske vid. Selsk. Skr. Mat.
  Nat. Kl.  \textbf{7},  1--22 (1906)

\bibitem{Thue12}
Thue, A.: {{\"U}ber} die gegenseitige {Lage} gleicher {Teile} gewisser
  {Zeichenreihen}. Norske vid. Selsk. Skr. Mat. Nat. Kl.  \textbf{1},  1--67
  (1912)

\end{thebibliography}
\end{document}